\documentclass[a4paper,twoside]{amsart}
\usepackage{amsmath}
\usepackage{amsfonts}
\usepackage{amssymb}
\usepackage{amsthm}
\usepackage{newlfont}
\usepackage{graphicx}
\usepackage{amscd}
\usepackage{young}

\theoremstyle{plain}
\newtheorem{Th}{Theorem}[section]
\newtheorem{Cor}[Th]{Corollary}

\newtheorem{Prop}[Th]{Proposition}
\theoremstyle{definition}
\newtheorem{Def}{Definition}[section]

\theoremstyle{remark}
\newtheorem*{Rem}{Remark}
\numberwithin{equation}{section}

\newcommand{\ZZ}{{\mathbb Z}}
\newcommand{\VV}{{\mathbb V}}

\newcommand{\bphi}{\boldsymbol{\phi}}

\newcommand{\bpsi}{\boldsymbol{\psi}}

\begin{document}

\title
{Non-autonomous multidimensional Toda system and multiple interpolation problem}

\author{Adam Doliwa}

\address{A. Doliwa, Faculty of Mathematics and Computer Science\\
	University of Warmia and Mazury in Olsztyn\\
	ul.~S{\l}oneczna~54\\ 10-710~Olsztyn\\ Poland} 
\email{doliwa@matman.uwm.edu.pl}

\date{}

\begin{abstract}
We study the interpolation analogue of the Hermite--Pad\'e type~I approximation problem. We provide its determinant solution and we write down the corresponding integrable discrete system as an admissible reduction of Hirota's discrete Kadomtsev--Petviashvili equations. Apart from the $\tau$-function form of the system we provide its variant, which in the simplest case of dimension two reduces to the non-autonomous discrete-time Toda equations. 
\end{abstract}
\maketitle

\section{Introduction}
The methods of the theory of integrable systems, originally discovered in the context of nonlinear waves~\cite{GGKM}, proliferated on many parts of contemporary mathematics. At the same time as theory embraced new domains, its links with the classical fields of mathematics were being revealed. 
The present work discusses the connection of the integrability with interpolation theory, considered usually as a part of the applied numerical analysis. Our main motivation was however the well known relation between the Pad\'{e} approximants, the Toda lattice equations, theory of orthogonal polynomials, Painlev\'{e} equations and random matrices~\cite{Brezinski-PTA-OP,BultheelvanBarel,Hirota-2dT,PGR-LMP,Adler-vanMoerbeke-1995,IDS}. 

\subsection{The Hermite--Pad\'{e} approximation and integrability}
Approximation by rational functions that give the best possible matching with the given expansion of the function into the Taylor series (the so-called Pad\'{e} approximation) is widely used in numerical calculations. It works especially well for functions that have singularities. The known determinant formulas, originating from Jacobi, are not numerically effective. In practice, recursive relationships (the Frobenius identities) between numerators and denominators of the approximants organized into two-dimensional arrays are used~\cite{BakerGraves-Morris}. These recursions were rediscovered after 100 years in the broader context of integrable systems theory~\cite{Hirota-2dT,IDS} and their usefulness in approximation problems can be explained by the existence of an (infinite) number of symmetries that stabilize the calculations. 

Due to the structure of determinants involved, the Pad\'{e} approximants are closely related to the theory of orthogonal polynomials, the importance of which in mathematics and theoretical physics cannot be overestimated~\cite{Szego,Geronimus,Ismail,VilenkinKlimyk}. Multiple orthogonal polynomials generalize the notion of  polynomials orthogonal with respect to one measure to polynomials satisfying the orthogonality conditions with respect to a system of several measures~\cite{Aptekarev,NikishinSorokin,VanAsche}. They have found applications in random matrix theory,  stochastic processes and in combinatorics~\cite{BleherKuijlaars,AptekarevKuijlaars,AdlervanMoerbeke,Kuijlaars}. Their theory can be traced back however to works of Hermite (it is worth noting that Pad\'{e} was his PhD student) on transcendence of Euler's number~$e$~\cite{Hermite,Hermite-P}. The corresponding generalization (for given several Taylor series) of the approximation, and thus of the multiple orthogonality, splits into two dual~\cite{Mahler-P} problems: the Hermite--Pad\'{e} type I and type II approximation.  

The question of the recurrence relations generalizing those of the Pad\'{e} approximation/orthogonal polynomials theory was raised within the numerical algorithms community~\cite{Paszkowski,DD-DC-2} and in works on multiple orthogonal polynomials in~\cite{AptekarevDerevyaginVanAssche,AptekarevDerevyaginMikiVanAssche}, where the time variable (continuous or discrete) shows up from an appropriate variation of the measure. 
In particular, in \cite{AdlervanMoerbekeVanhaecke} it was shown that the determinants of the moment matrices satisfy, upon adding one set of ``time'' deformations for each weight, the multi-component Kadomtsev--Petviashvili (KP) hierarchy~\cite{DKJM,Kac-vLeur}. The duality relations between type I and type II multiple orthogonal polynomials have been formulated there as a bilinear identity satisfied by the Riemann--Hilbert matrix and its adjoint. Generalization of such an approach to the multi-component 2D Toda hierarchy \cite{UenoTakasaki} within the context of the Gauss--Borel factorization was given in~\cite{Alvarez-FernandezPrietoManas}. 
Among other related 
topics discussed there one can find the construction of discrete flows in terms of Miwa transformations~\cite{Miwa}.  

As it was clarified in~\cite{Doliwa-Siemaszko-HP} the difference equations of the Hermite--Pad\'{e} approximation theory are equivalent (under simple linear change of independent variables) to the $\tau$-function form of the integrable multidimensional generalization of the discrete-time Toda lattice given in~\cite{AptekarevDerevyaginMikiVanAssche} and of its linear problem. It was shown there also that they can be obtained as an admissible reduction of Hirota's discrete KP system~\cite{Hirota,Miwa}, which plays a special role within the theory of integrable equations and their applications, see reviews~\cite{KNS-rev,Zabrodin}. It should be mentioned that ``multidimensionality'' refers here to the number of independent (discrete) variables which can be arbitrarily large, but the initial boundary value problem is typical for two-dimensional systems. Other integrable systems related to the Hermite--Pad\'{e} approximation have been derived and studied in~\cite{ManoTsuda,NagaoYamada}.


\subsection{The rational interpolation problem and formulation of the main results}
The analogous problem of rational interpolation (or the Cauchy interpolation~\cite{Cauchy,CuytWuytack})  is that of finding rational function with prescribed degrees of the nominator and the denominator, which assumes prescribed values at given distinct points (nodes of the interpolation).
Its solution in terms of certain determinants was given by Jacobi~\cite{Jacobi}; see also \cite{GathenGerhard} for discussion of computational complexity of the Cauchy interpolation and its relation to the Extended Euclidean Algorithm and the Chinese Remainder Theorem. 
If instead of prescribing the values at distinct points we take the confluent limit when the points coincide, and consecutively we prescribe an initial segment of the Taylor expansion of an analytic function we obtain the Pad\'{e} approximation problem (various intermediate cases are possible as well~\cite{CuytWuytack}).

The analogous recurrence relations behind the rational interpolation are not so well studied, see however the numerical analysis literature \cite{BultheelvanBarel,CuytWuytack,BulirschStoer} for the non-autonomous versions of the Frobenius identities where the interpolation nodes appear explicitly. One can find there also the non-autonomous version of the Wynn recurrence \cite{Claessens,Wynn} (the missing identity of Frobenius) whose integrability/multidimensional consistency was recently studied in \cite{Kels} within broader context of cross-shaped difference equations. From the other side, a non-autonomous version of the discrete time Toda lattice equation~\cite{SpiridonovZhedanov,Hirota-naTl,SpiridonovZhedanov-2,MukaihiraTsujimoto,KajiwaraMukaihira,MaedaTsujimoto} was studied recently from the point of view of the discrete orthogonal polynomials, integrability and linear algebra algorithms. One of results of the present work is demonstration of the connection between these problems. 

However the main results of the paper are about the relation between integrability and the interpolation generalization of the Hermite--Pad\'{e} approximation problem of type~I, which we call the multiple interpolation problem.

It should be mentioned that a part of difference equations considered here was investigated by Mahler~\cite{Mahler-P} and used in~\cite{NagaoYamada} in relation to the discrete Painlev\'{e} equations. In particular
\begin{itemize}
	\item we present a solution of the multiple interpolation problem in terms of certain determinants built from of the interpolation data;
	\item we derive non-autonomous relations between the determinants generalizing the corresponding equations~\cite{Paszkowski} known from the Hermite--Pad\'{e} approximation theory;
	\item we show integrability of the resulting equations which form a non-autonomous generalization of the multidimensional Toda lattice equations;
	\item we demonstrate that the equations can be obtained as an admissible reduction of the Hirota system;
	\item we give another form of the equations which we expect to be more relevant in the theory of generalized orthogonal polynomials;
	\item finally, we show how in dimension two, which contains as a special case the rational interpolation of a single function, the above results reproduce both the non-autonomous discrete-time Toda system and the non-autonomous Wynn recurrence.
\end{itemize}
Because the generalization from multiple approximation to the multiple interpolation preserves the integrability structure of the underlying equations we are convinced that the above results will be relevant in all the fields where  the multiple approximation or the closely related multiple orthogonality techniques have found already their application.

The structure of the paper is as follows. In the rest of this  Section we recall basic ingredients of the Hirota system and its multidimensional Toda lattice reduction. Then in Section~\ref{sec:mCi} we present the interpolation generalization of the Hermite--Pad\'{e} type I problem and we give its determinant solution. In the next Section~\ref{sec:DE}, basing on certain determinant identities, we show that the polynomial solutions of the interpolation problem satisfy difference equations with discrete variables being the degrees of the polynomials. Section~\ref{sec:namTs} is devoted to discussion of integrability of the equations. Finally, in Section~\ref{sec:2D} we show how in the simplest case we recover the known results on the non-autonomous discrete-time Toda lattice and the corresponding generalization of the Wynn recurrence.  We close the paper by presenting conclusions and related open problems.


\subsection{The Hirota equation and multidimensional Toda system} \label{sec:Hirota}
Let $\ZZ^m$ be $m\geq 3$-dimensional integer lattice with $n=(n_1, \dots , n_m) = \sum_{j=1}^m n_j e_j$ being the discrete variable, and $(e_j)_{j=1}^m$ being the standard basis. Hirota's discrete KP equation \cite{Hirota,Miwa} reads as follows
\begin{equation} \label{eq:H-M}
\tau(n+e_i)\tau(n+e_j+e_k) - \tau(n+e_j)\tau(n+e_i+e_k) + \tau(n+e_k)\tau(n+e_i+e_j) =0, 
\end{equation}
where $\tau\colon\ZZ^m\to \Bbbk$ is an unknown function with values in a field $\Bbbk$ (usually the real or complex numbers, but see~\cite{BialeckiDoliwa} for finite field solutions or \cite{FWN-Capel,Nimmo-NCKP} for a non-commutative version of the system); here also $1\leq i< j <k \leq m$. In order to apply techniques of integrable systems theory~\cite{IDS} it is important to represent the non-linear system~\eqref{eq:H-M} as compatibility condition of the corresponding linear problem which we consider the following form~\cite{DJM-II,Saito-Saitoh}
\begin{equation}
	\label{eq:ad-lin-bil}
	\bphi(n+e_j) \tau(n+e_i)  -  \bphi(n+e_i)\tau (n+e_j) =  \bphi(n+e_i+e_j) \tau(n), \qquad i<j,
	\end{equation}
where $\bphi\colon\ZZ^m\to\VV$, called the wave function, takes values in a linear space over $\Bbbk$.

Adding to the linear problem \eqref{eq:ad-lin-bil} the constraint
\begin{equation} \label{eq:P-lin}
x\bphi (n,x) \tau(n) = \bphi(n+e_1,x) \tau(n-e_1) + \dots + \bphi(n+e_m,x) \tau(n-e_m),
\end{equation}
results~\cite{Doliwa-Siemaszko-HP}, by compatibility, in supplementing the Hirota system~\eqref{eq:H-M} by the equation
\begin{equation}
\label{eq:tau-Paszkowski}
\tau(n)^2 = \tau(n+e_1) \tau(n-e_1) + \dots + \tau(n+e_m) \tau(n-e_m).  
\end{equation}
Equations \eqref{eq:H-M} and \eqref{eq:tau-Paszkowski} together are known as the multidimensional Toda system, whose special solutions are relevant in the theory of the Hermite--Pad\'{e} approximation problem~\cite{Doliwa-Siemaszko-HP} or in the theory of multiple orthogonal polynomials~\cite{AptekarevDerevyaginMikiVanAssche}.

The linear system~\eqref{eq:ad-lin-bil}-\eqref{eq:P-lin} is meaningful also in the special case $m=2$
\begin{align}
	\label{eq:ad-lin-bil-m2}
  \bphi(n+e_1+e_2,x) \tau(n) & =	\bphi(n+e_2,x) \tau(n+e_1)  -  \bphi(n+e_1,x)\tau (n+e_2) , \\
\label{eq:P-lin-m2}
x\bphi (n,x) \tau(n)  & = \bphi(n+e_1,x) \tau(n-e_1) +  \bphi(n+e_2,x) \tau(n-e_2),	
	\end{align}
and results in the single equation
	\begin{equation}
\label{eq:Toda}
\tau(n)^2 = \tau(n+e_1) \tau(n-e_1) + \tau(n+e_2) \tau(n-e_2),
	\end{equation}
known in the theory of integrable systems as the discrete-time Toda lattice equation~\cite{Hirota-2dT}. The system~\eqref{eq:ad-lin-bil-m2}-\eqref{eq:Toda} forms a part of the so called Frobenius identities in the Pad\'{e} approximation problem~\cite{BakerGraves-Morris}, and special determinant solutions of the system provide solution to the problem.

\begin{Rem}
The auxiliary variable $x$ is known in the soliton theory as the spectral parameter. It does not show up in the linear problem of the Hirota system before imposing the reduction constraint~\eqref{eq:P-lin}, but even there it plays an important role in construction of its solutions  by the algebro-geometric techniques~\cite{Krichever} or by the non-local $\bar\partial$-dressing method~\cite{Dol-Des}. 
\end{Rem}

\section{Multiple interpolation problem and its determinant solution} \label{sec:mCi}
Given $\Bbbk$-valued functions $(f_1,\dots , f_m)$ of single variable $x$ and given a sequence $(x_s)_{s=1,2,\dots}$. Consider $n=(n_1,\dots , n_m) = \sum_{k=1}^m n_k e_k$, an element of $\ZZ_{\geq -1}^m$, where we also write $|n|= n_1 + \dots + n_m$. 
\begin{Def}
By a~\emph{multiple interpolation form of degree $n$} we call any system of polynomials $(Y_1, \dots , Y_m)$ in $\Bbbk[x]$, not all equal to zero, with corresponding degrees $\deg Y_i \leq n_i$, $i=1,\dots , m$ (degree of the zero polynomial by definition equals $-1$), and such that 
\begin{equation} \label{eq:mCi}
Y_1(x_s)f_1(x_s) + \dots + Y_m(x_s) f_m(x_s) = 0 \qquad s = 1, 2 , \dots , |n|+m-1.
\end{equation}
\end{Def}
\begin{Rem}
	In the confluent case when all the points $x_s$ coincide, and with transition to the appropriate tangency condition, the above multiple interpolation problem becomes the Hermite--Pad\'{e} approximation problem of type I \cite{Mahler-P,Aptekarev}. 
\end{Rem}

Define matrix $\mathcal{M}(n)$ of $N=(|n|+m-1)$ rows and $(|n|+m)$ columns 
\begin{small}
\begin{equation}
\mathcal{M}(n) = \left( \begin{matrix}
f_1(x_1) & x_1f_1(x_1) &  \cdots   & x_1^{n_1}f_1(x_1) & \cdots & \cdots & f_m(x_1) & x_1 f_m(x_1) & \cdots  & x_1^{n_m} f_m(x_1) \\
f_1(x_2) & x_2 f_1(x_2) & \cdots    & x_2^{n_1}f_1(x_2)   & \cdots & \cdots & f_m(x_2) & x_2 f_m(x_2) & \cdots  & x_2^{n_m} f_m(x_2) \\
\vdots & \vdots   & \ddots &  \vdots & \cdots  & \cdots & \vdots &  \vdots  & \ddots  &  \vdots  \\
f_1(x_{N}) & x_N f_1(x_{N})  & \cdots & x_N^{n_1}f_1(x_N) & \cdots & \cdots & f_m(x_N)& x_N f_m(x_N) & \cdots & x_N^{n_m} f_m(x_N)
\end{matrix} \right) ;
\end{equation}
\end{small}
its columns are divided into $m$ (possibly empty) groups, the $i$th group is composed out of $n_i +1$ columns depending on the values $f_i(x_s)$ only, $s=1, \dots , N$.
Let us supplement $\mathcal{M}(n)$ at the bottom by the line
\begin{equation*}
b(n,x) = \left( f_1(x), xf_1(x), \dots , x^{n_1} f_1(x), \cdots \, \cdots , f_m(x), x f_m(x), \dots , x^{n_m} f_m(x) \right)
\end{equation*}
and denote by $\mathcal{D}(n,x)$ the determinant of the resulting square matrix. Similarly, supplement the matrix $\mathcal{M}(n)$ at the bottom by the line 
\begin{equation}
\label{eq:Xk}
X_k(n,x) = (0, \dots, 0, \dots \; \dots , 1, x , \dots , x^{n_k}, \dots \; \dots , 0, \dots ,0), \quad  k=1,\dots m , 
\end{equation}
consisting of zeros except for the $k$th block of the form $1, x, \dots , x^{n_k}$. Its determinant $Z_k(n,x)$, given explicitly as 
\begin{equation} \label{eq:Z-det}
Z_k(n,x) = \left|
\begin{smallmatrix}
f_1(x_1) &  x_1f(x_1) & \cdots   & x_1^{n_1}f_1(x_1) & \cdots & \cdots & f_k(x_1) & x_1 f_k(x_1) & \cdots & x_1^{n_k}f_k(x_1) &\cdots & \cdots & f_m(x_1) &  \cdots  & x_1^{n_m} f_m(x_1) \\
f_1(x_2) & x_2 f(x_2) &  \cdots   & x_2^{n_1}f_1(x_2) & \cdots & \cdots & f_k(x_2) & x_2 f_k(x_2) & \cdots  & x_2^{n_k}f_k(x_2) &\cdots & \cdots & f_m(x_2) & \cdots  & x_2^{n_m} f_m(x_2) \\
\vdots & \vdots &  \ddots  &  \vdots & \cdots &\cdots & \vdots& \vdots & \ddots & \vdots &&   & \vdots &  & \vdots\\
f_1(x_N) & x_N f(x_N) & \cdots & x_N^{n_1} f_1(x_N) & & &f_k(x_N)& x_N f_k(x_N) &\cdots & x_N^{n_k} f_k(x_N)& & &  f_m(x_N) & \cdots & x_N^{n_m} f_m(x_N) \\
0 & 0 & \cdots & 0 & \cdots & \cdots &1 & x & \cdots & x^{n_k} & \cdots &\cdots & 0 & \cdots & 0
\end{smallmatrix} 
\right|,
\end{equation}
is a polynomial of degree not exceeding $n_k$.
\begin{Prop}
	The polynomials $(Z_1(n,x),\dots , Z_m(n,x))$ provide solution of the above multiple interpolation problem~\eqref{eq:mCi}.
\end{Prop}
\begin{proof}
	The Laplace expansion of $\mathcal{D}(n,x)$ with respect to the last row gives
	\begin{equation} \label{eq:Z-f-Delta}
	\mathcal{D}(n,x) = Z_1(n,x) f_1(x) + \dots + Z_m(n,x) f_m(x) .
	\end{equation}
From the other side, by elementary properties of the determinants we have $\mathcal{D}(n,x_s) = 0$, for all $s=1, \dots , N$.
\end{proof}
\begin{Rem}
The above polynomials will be called \emph{the canonical multiple interpolation form of degree $n$}.
\end{Rem}
By $\Delta(n) = \mathcal{D}(n,x_{N+1})$ denote the determinant 
\begin{equation} \label{eq:Delta-det}
\Delta(n) = \left|
\begin{smallmatrix}
f_1(x_1) &  x_1f(x_1) & \cdots   & x_1^{n_1}f_1(x_1) & \cdots & \cdots  & f_m(x_1) & x_1 f_m(x_1) & \cdots  & x_1^{n_m} f_m(x_1) \\
f_1(x_2) & x_2 f(x_2) &  \cdots   & x_2^{n_1}f_1(x_2) & \cdots & \cdots  & f_m(x_2) & x_2 f_m(x_2) & \cdots  & x_2^{n_m} f_m(x_2) \\
\vdots & \vdots &  \ddots  &  \vdots & \cdots & \cdots & \vdots & \vdots & \ddots  & \vdots\\
f_1(x_N) & x_N f(x_N) & \cdots & x_N^{n_1} f_1(x_N)  &  &  & f_m(x_N) & x_N f_m(x_N) & \cdots & x_N^{n_m} f_m(x_N) \\
f_1(x_{N+1}) & x_{N+1} f(x_{N+1}) & \cdots & x_{N+1}^{n_1} f_1(x_{N+1}) & \cdots & \cdots &  f_m(x_{N+1}) & x_{N+1} f_m (x_{N+1}) &\cdots & x_{N+1}^{n_m} f_m(x_{N+1}) 
\end{smallmatrix} 
\right|
\end{equation}
of the matrix $\mathcal{M}(n)$ supplemented by the line
\begin{equation} \label{eq:bN}
b(n,x_{N+1}) = \left( f_1(x_{N+1}),  \dots , x_{N+1}^{n_1} f_1(x_{N+1}), \cdots \; \cdots , f_m(x_{N+1}), \dots , x_{N+1}^{n_m} f_m(x_{N+1}) \right),
\end{equation} 
as the last row. 
\begin{Cor}
By the Laplace expansion of the determinant \eqref{eq:Z-det} we obtain the leading term of the polynomial $Z_k(n,x)$ which reads
\begin{equation} \label{eq:leading-term-Z}
Z_k(n,x) = (-1)^{(n_{k+1} + \dots + n_m) + m - k}\Delta(n-e_k) x^{n_k} + \text{lower order terms}.
\end{equation}
\end{Cor}
Therefore if for all $n$ the determinants $\Delta(n)$ do not vanish then the polynomials $Z_k$ are of the maximal order. Such a system of functions $(f_1, \dots , f_m)$ and the sequence $(x_s)_{s=1,2,\dots}$ will be called \emph{perfect} in analogy with the standard terminology~\cite{Mahler-P}. Then for each $n$ the space of multiple interpolation forms is one-dimensional. 

\section{Difference equations behind the multiple interpolation problem} \label{sec:DE}
\subsection{Multiple interpolants and the Hirota system}
\label{sec:CH-equations}
In this Section we study integrable equations satisfied by the polynomials $Z_\ell(n,x)$ and the function $\Delta(n)$. 
\begin{Prop}
	The polynomials $Z_\ell(n,x)$, $\ell = 1, \dots , m$, satisfy the following system of linear equations
\begin{equation} \label{eq:Z-Delta}
Z_{\ell}(n+e_i + e_j,x) \Delta(n) = Z_{\ell}(n+e_j,x) \Delta(n+e_i) - 
Z_{\ell}(n+e_i,x) \Delta(n+e_j), \qquad 1\leq i<j \leq m.
\end{equation}
\end{Prop}
\begin{proof}
	Apply to the determinant $Z_\ell(n+e_i+e_j,x)$ the Jacobi identity~\cite{Hirota-book} (known also as the Sylvester identity or the Dodgson condensation rule) with respect to the last two rows and the last columns of the blocks $i$ and $j$. 
\end{proof}
\begin{Rem}
In the non-degenerate situation one can apply yet another technique. The systems of po\-ly\-no\-mials on both sides of \eqref{eq:Z-Delta} are multiple interpolation forms of the same order, thus must be proportional. The coefficient of proportionality can be found by comparison of the leading order term for the polynomials  $Z_i$ on both sides. 
\end{Rem}
The above linear system is identical with the linear problem~\eqref{eq:ad-lin-bil} of the Hirota equation~\eqref{eq:H-M}; we identify $\Delta$ with the $\tau$-function and the vector $(Z_1,\dots ,Z_m)$ as the wave function $\bphi$.
By eliminating either $Z_\ell$ or $\Delta$ from equations \eqref{eq:Z-Delta} for three pairs of indices $i<j$, $i<k$ and $j<k$ we obtain the well known nonlinear integrable equations. Staying on the level of the interpolation problem we will show them using the tools of determinant identities.
\begin{Prop}
	The basic determinant $\Delta$ and the canonical multiple interpolation forms $Z_\ell$ satisfy the standard bilinear Hirota equations
	\begin{gather}
	\label{eq:H-Delta} \; \:
	\Delta(n+e_i + e_j) \Delta(n+e_k) - \Delta(n+e_i + e_k) \Delta(n+e_j) +
	\Delta(n+e_j + e_k) \Delta(n+e_i) = 0, \\
	\label{eq:H-Z}
	Z_{\ell}(n+e_i + e_j)Z_{\ell}(n+e_k) - Z_{\ell}(n+e_i + e_k)Z_{\ell}(n+e_j)  + Z_{\ell}(n+e_j + e_k)Z_{\ell}(n+e_i)  =  0,
	\end{gather} 
	where $1\leq i < j < k \leq m$, and $\ell = 1,\dots , m$.
\end{Prop}
\begin{proof}
	Equations \eqref{eq:H-Delta}-\eqref{eq:H-Z} can be obtained directly from the Jacobi identities applied to corresponding determinants. For equation \eqref{eq:H-Delta} we take determinant of the matrix $\mathcal{M}(n+e_i + e_j + e_k)$ supplemented at the bottom by the row with all zeros except one at the last column of the $k$th group, and apply the Jacobi identity with respect to last two rows and the last columns of the blocks $i$ and $j$. 
	
	In order to prove directly equation \eqref{eq:H-Z} we consider first the case $\ell \neq i,j,k$. Consider determinant of the matrix $\mathcal{M}(n+e_i + e_j + e_k)$ with the last row replaced by $X_\ell$, and supplemented at the bottom by the row with all zeros except one at the last column of the $k$th group. Then apply the Jacobi identity with respect to the third from the bottom and the last rows, and the last columns of the blocks $i$ and~$j$. In the case when $\ell$ equals one of $i,j,k$, instead of $X_\ell$ we use the corresponding row $X$, of appropriate dimension.
\end{proof}
\begin{Rem}
	In the context of the multiple interpolation problem  and in the non-degenerate case equation \eqref{eq:H-Delta} can be also obtained from~\eqref{eq:Z-Delta} by inserting the leading term coefficient form~\eqref{eq:leading-term-Z} of the polynomial $Z_k$ and collecting the higher order terms. 
\end{Rem}

\subsection{The admissible constraint}
In application to Hermite--Pad\'{e} approximation problem equations \eqref{eq:Z-Delta} and \eqref{eq:H-Delta} were supplemented in~\cite{Paszkowski} by additional constraints (equations~\eqref{eq:tau-Paszkowski} and \eqref{eq:P-lin} in the $\Delta-Z$ notation); see also \cite{BakerGraves-Morris}.  Let us find their analogs in the context of the multiple interpolation problem.
\begin{Prop} \label{prop:Paszkowski-CH}
	The canonical multiple interpolation forms satisfy the constraint
	\begin{equation}
	\label{eq:Paszkowski-CH}
	(x - x_{|n|+m}) {Z_\ell}(n,x) \Delta(n) = Z_{\ell}(n+e_1,x) \Delta(n-e_1) + \dots + Z_{\ell} (n+e_m,x) \Delta(n-e_m),  
	\end{equation}
	where $\ell = 1, \dots , m$.
\end{Prop}
\begin{Rem}
	Actually, equation of a form similar to~\eqref{eq:Paszkowski-CH} can be found in Mahler's paper~\cite{Mahler-P} (the first equation of $\S~24$), but the important difference was caused by the lack of the determinant interpretation what made some proofs applicable in the generic situation only.
\end{Rem}
\begin{proof} We will present two proofs. The first one works under the non-degeneracy assumption and is based on the corresponding ideas of \cite{Mahler-P}, while the second proof exploits determinant identities.
	
I. The system of $m$ polynomials 
\begin{equation*}
(x - x_{|n|+m}) {Z_\ell}(n,x) \Delta(n) - Z_{\ell}(n+e_1,x) \Delta(n-e_1) - \dots - Z_{\ell} (n+e_{m-1},x) \Delta(n-e_{m-1})
\end{equation*}	
forms a solution of the multiple interpolation problem. By the examination of the leading terms~\eqref{eq:leading-term-Z} of the polynomials the system has the same degrees as the system $Z_\ell(n+e_m,x)$, thus they must be proportional. The corresponding coefficient can be found by comparison of the highest order terms of the $m$th polynomials on both sides.

II.  We will use abbreviated variant of notation introduced in Section~\ref{sec:mCi}, in particular
\begin{equation}
Z_\ell = \left| \begin{matrix} \mathcal{M} \\X_\ell \end{matrix} \right| , \qquad \Delta = \left| \begin{matrix} \mathcal{M} \\ b_{N+1} \end{matrix} \right|, \qquad  N=|n| + m -1. 
\end{equation}
By $\mathcal{M}^+$ denote the matrix $\mathcal{M}$ with each $j$th row multiplied by $x_j$, $j=1,\dots , N$, similarly $b_{N+1}^+ = x_{N+1}b_{N+1}$. Consider the determinant 
\begin{equation*}
\mathcal{D}_\ell = \left| \begin{array}{c|c} \mathcal{M} & \mathcal{M}^+ \\
b_{N+1} & b_{N+1}^+ \\ \hline
X_\ell & x\, X_\ell \\
0^N_{N+1} & \mathcal{M}
\end{array} \right| ,
\end{equation*}
where  $0^N_{N+1}$ is $N\times (N+1)$ array of zeros. By using elementary row operations we convert the upper right block occupied by $\mathcal{M}^+$ to a block of zeros. Then by applying the generalized Laplace expansion with respect to the first $(N+1)$ columns we obtain
\begin{equation*}
\mathcal{D}_\ell = \left| \begin{array}{c|c} 
\begin{matrix}
\mathcal{M} \\ b_{N+1} 
\end{matrix} & 0^{N+1}_{N+1} \\ \hline
0^{N+1}_{N+1} & \begin{matrix}
x\, X_\ell \\ \mathcal{M}
\end{matrix}
\end{array} \right| + 
\left| \begin{array}{c|c} 
\mathcal{M} & 0^N_{N+1}  \\
0_{N+1} & b_{N+1}^+ \\ \hline
X_\ell & 0_{N+1} \\
0^N_{N+1} & \mathcal{M}
\end{array} \right| = (-1)^N (x-x_{N+1}) Z_\ell \Delta,
\end{equation*}
where $0_{N+1}$ is a row vector with $N+1$ zeros.

From the other hand, by elementary column operations we can create zeros in all of the top $(N+2)$ entries in every one of the columns in the right half of $\mathcal{D}_\ell$ except for the last ones in each of $m$ blocks. Then we split the resulting determinant into $m$ ones with such a single column only (with remaining $N$ entries equal to zero). Finally, by moving the column left to the end of the corresponding block, which results in the factor $(-1)^N$, we obtain from each such determinant the appropriate summand in equation~\eqref{eq:Paszkowski-CH}.
\end{proof}
\begin{Prop} \label{prop:Delta-Paszkowski-CH}
	The determinants $\Delta$ satisfy the constraint
	\begin{equation}
	\label{eq:Delta-Paszkowski-CH}
	(x_{|n|+m+1} - x_{|n|+m}) \tilde{\Delta}(n) \Delta(n) = \Delta(n+e_1) \Delta(n-e_1) + \dots + \Delta (n+e_m) \Delta(n-e_m),  
	\end{equation}
	where
$\tilde{\Delta}(n)=\mathcal{D}(n,x_{|n|+m+1})$ is the determinant of the matrix $\mathcal{M}(n)$ supplemented by the line
\begin{equation} \label{eq:bN2}
b(n,x_{N+2}) = \left( f_1(x_{N+2}),  \dots , x_{N+2}^{n_1} f_1(x_{N+2}), \cdots \; \cdots , f_m(x_{N+2}), \dots , x_{N+2}^{n_m} f_m(x_{N+2}) \right)
\end{equation} 
at the bottom. Such determinants satisfy equations
\begin{equation} \label{eq:D-tD}
\Delta(n+e_i + e_j) \Delta(n) = \Delta(n+e_i) \tilde{\Delta}(n+e_j) - \Delta(n+e_j) \tilde{\Delta}(n+e_i),
\qquad 1\leq i < j \leq m.
\end{equation}
\end{Prop}
\begin{proof} We will give again two proofs (see the next Section for the third one).
	
I. By equation \eqref{eq:Z-f-Delta} the constraint~\eqref{eq:Paszkowski-CH} leads to 
\begin{equation*}
(x - x_{|n|+m}) \mathcal{D}(n,x) \Delta(n) = \mathcal{D}(n+e_1,x) \Delta(n-e_1) + \dots + \mathcal{D} (n+e_m,x) \Delta(n-e_m),
\end{equation*}
which evaluated at $x=x_{|n|+m+1}$ gives the first part of the statement. The second part is a consequence of the same technique as above, but applied to equation~\eqref{eq:Z-Delta} and the argument $x=x_{|n|+m+2}$.

II. Equation~\eqref{eq:Delta-Paszkowski-CH} can be obtained by using the same determinant identities as in the proof of Proposition~\ref{prop:Paszkowski-CH} but applied to the determinant 
\begin{equation*}
 \left| \begin{array}{c|c} \mathcal{M} & \mathcal{M}^+ \\
b_{N+1} & b_{N+1}^+ \\ \hline
b_{N+2} & b_{N+2}^+ \\
0^N_{N+1} & \mathcal{M}
\end{array} \right| .
\end{equation*}
Equation \eqref{eq:D-tD} is simpler and can be directly obtained by application of the Jacobi identity with respect to the last two rows and the last columns of the $i$th and $j$th blocks to the determinant \eqref{eq:Delta-det} representing $\Delta(n+e_i+e_j)$.
\end{proof}

\section{The non-autonomous multidimensional Toda system} \label{sec:namTs}

\subsection{The $\tau$-function formulation of the equations}
Let us provide the third proof of Proposition~\ref{prop:Delta-Paszkowski-CH},  reformulated within more general context of integrable systems theory, which we keep up to the end of the paper; we replace therefore the $\Delta-Z$ notation of the interpolation theory into the $\tau-\bphi$ notation of Section~\ref{sec:Hirota}. In particular, the discrete variable $n$ belongs to the $m$ dimensional integer lattice $\ZZ^m$. Consequently, the range of the index $s$ of the sequence $(x_s)$ is increased from natural to the integer numbers $\ZZ$.
\begin{Prop}
Adding to the linear problem \eqref{eq:ad-lin-bil} the constraint
\begin{equation} \label{eq:int-P-lin}
(x-x_{|n|+m})\bphi (n,x) \tau(n) = \bphi(n+e_1,x) \tau(n-e_1) + \dots + \bphi(n+e_m,x) \tau(n-e_m),
\end{equation}
results, by compatibility, in supplementing the Hirota system~\eqref{eq:H-M} by the equation
	\begin{equation}
\label{eq:tau-Paszkowski-CH}
(x_{|n|+m+1} - x_{|n|+m}) \tilde{\tau}(n) \tau(n) = \tau(n+e_1) \tau(n-e_1) + \dots + \tau(n+e_m) \tau(n-e_m),  
\end{equation}
where the additional function $\tilde{\tau}$ satisfies the system
\begin{equation} \label{eq:t-tau}
\tau(n+e_i + e_j) \tau(n) = \tau(n+e_i) \tilde{\tau}(n+e_j) - \tau(n+e_j) \tilde{\tau}(n+e_i),
\qquad 1\leq i < j \leq m.
\end{equation}
\end{Prop}
\begin{proof}
Notice first, that the linear system~\eqref{eq:ad-lin-bil} and \eqref{eq:tau-Paszkowski-CH} imply the identity
\begin{equation*} 
(x-x_{|n|+m}) \bphi(n,x) \tau(n) \tau(n+e_i) =
\bphi(n+e_i,x) S(n) +  \sum_{j=1}^m \mathrm{sgn}(j-i) \bphi(n+e_i+e_j,x) \tau(n) \tau(n-e_j), 
\end{equation*}	
where
\begin{equation}
S(n) = \tau(n+e_1) \tau(n-e_1) + \dots + \tau(n+e_m)\tau(n-e_m).
\end{equation} Then by applying the Hirota system and the above identity in two instances for shifts back in the $i$th and $j$th discrete variables we get
\begin{equation*} \begin{split}
(x-x_{|n|+m-1})  \left( \bphi(n-e_i,x) \tau(n-e_j) 
-\bphi(n-e_j)\tau(n-e_i) \right) \tau(n) = & \\ =\bphi(n,x) \tau(n-e_i)\tau(n-e_j) \left( \frac{S(n-e_i)}{\tau^2(n-e_i)} -
\frac{S(n-e_j)}{\tau^2(n-e_j)} \right) + \tau(n-e_i - e_j)
  \sum_{k=1}^m & \bphi(n+e_k) \tau(n-e_k).
\end{split}
\end{equation*}
Using then the linear equations once again we obtain finally
\begin{equation*} \begin{split}
(x-x_{|n|+m-1})  \bphi(n,x) \tau(n-e_i - e_j) & \tau(n) =
\bphi(n,x) \tau(n-e_i)\tau(n-e_j) \left( \frac{S(n-e_i)}{\tau^2(n-e_i)} -
\frac{S(n-e_j)}{\tau^2(n-e_j)} \right) + \\ &+
(x-x_{|n|+m})  \bphi(n,x) \tau(n-e_i - e_j) \tau(n). \end{split}
\end{equation*}
By the standard comparison of the terms in front of different powers of the (spectral) parameter $x$ we get the non-linear equations written exclusively in terms of the $\tau$-function
\begin{equation} \label{eq:nlin-tau-S}
x_{|n|+m}-x_{|n|+m-1} = \frac{ \tau(n-e_i)\tau(n- e_j)}
{\tau(n-e_i - e_j) \tau(n)} \left( \frac{S(n-e_i)}{\tau^2(n-e_i)} -
\frac{S(n-e_j)}{\tau^2(n-e_j)} \right), \qquad i<j.
\end{equation}
In order to transform the equations in the bilinear form let us introduce the function $\tilde{\tau}(n)$ by substitution
\begin{equation}
S(n) = \left( x_{|n|+m+1}-x_{|n|+m} \right)  \tilde{\tau}(n)\tau(n),
\end{equation}
what gives equation~\eqref{eq:tau-Paszkowski-CH}. Then the  
equations \eqref{eq:nlin-tau-S} reduce to \eqref{eq:t-tau}, what concludes the proof.
\end{proof}
\begin{Rem}
The presence of the auxiliary $\tau$-function~\cite{MukaihiraTsujimoto}, denoted here by $\tilde{\tau}$, was essential in construction of the soliton solutions~\cite{KajiwaraMukaihira} of the non-autonomous discrete-time Toda system. In the context of the multiple interpolation problem its presence is quite natural as a direct analog of the determinant $\tilde{\Delta}$.
\end{Rem}

\subsection{Another form of the  non-autonomous multidimensional discrete Toda system}
\label{sec:AB}

Let us distinguish the last direction $e_m$ and define functions $\bpsi(n,x) = \bphi(n,x)/ \tau(n-e_m)$, and then consider
\begin{equation} \label{eq:A-B-def}
A_i(n) = \frac{\tau(n+e_i - e_m) \tau(n-e_i)}{\tau(n) \tau(n-e_m)}, \qquad B_i(n) = 
	\frac{\tau(n+e_i - e_m) \tau(n+e_m)}{\tau(n) \tau(n+e_i)},
\quad i = 1, \dots , m-1.
\end{equation}
	The linear problem \eqref{eq:ad-lin-bil}, \eqref{eq:int-P-lin} in terms of function $\psi(n,x)$ reads as follows
\begin{align} \label{eq:ad-lin-bil-AB}
\bpsi(n+e_i + e_m, x) = & 
\bpsi(n+e_m, x) -  \bpsi(n+e_i, x)\,B_i(n) ,\quad i=1, \dots , m-1,\\ \label{eq:int-P-lin-AB}
(x - x_{|n|+m}) \bpsi(n, x) = & \bpsi(n+e_1, x) \,A_1(n) + \dots \bpsi(n+ e_{m-1}, x) \, A_{m-1}(n) + \bpsi(n+ e_m, x).
\end{align}
By taking compatibility conditions of the above linear system we obtain difference equations for the coefficients/functions $A_i(n)$ and $B_i(n)$. All these equations have their counterparts described in Proposition~\ref{prop:Paszkowski-CH} or follow directly from the definition of the coefficients in terms of the $\tau$-function. We expect however that this form will be important in applications of the system to, for example, theory of multiple orthogonal polynomials.
\begin{Prop}
	Taking only the first part \eqref{eq:ad-lin-bil-AB} of the linear system we can derive the following equations involving the shifts in $i\neq j$ directions
\begin{equation} \label{eq:lin-BiBj}
\bpsi(n+e_i + e_j, x)(B_i(n+e_j) - B_j(n+e_i)) =
\bpsi(n+e_i, x) B_i(n) -  \bpsi(n+e_j, x)\,B_j(n) ,
\end{equation}
where $ i,j=1, \dots , m-1$. The compatibility condition of the linear system reads
\begin{align}
\label{eq:BB}
B_i(n+e_m) B_j(n+e_i) & = B_j(n+e_m) B_i(n+e_j), \\
\label{eq:B+B}
B_i(n) + B_j(n+e_i) & = B_j(n) + B_i(n+e_j).
\end{align} 
\end{Prop}
\begin{proof}
Let us subtract from equation \eqref{eq:ad-lin-bil-AB} shifted in $j$th direction its version with the indices $i$ and $j$ reversed, what eliminates $\bpsi(n+e_i+e_j+e_m)$ and gives \eqref{eq:lin-BiBj}. Eliminating from the same equations $\bpsi(n+e_i+e_j)$ and shifting back in $m$-direction we obtain
\begin{equation*} \label{eq:lin-BiBj-m}
\bpsi(n+e_i + e_j, x)(B_i(n+e_j-e_m) - B_j(n+e_i-e_m)) =
\bpsi(n+e_i, x) B_i(n+e_j-e_m) -  \bpsi(n+e_j, x)\,B_j(n+e_i - e_m) ,
\end{equation*}
what compared with \eqref{eq:lin-BiBj} gives \eqref{eq:BB} and \eqref{eq:B+B}.
\end{proof}

The above equations are just another form of the discrete KP system. In what follows we discuss the impact of the constraint \eqref{eq:int-P-lin-AB}. We skip the detailed derivation of the corresponding equations because it follows closely the reasoning of the previous Section.
\begin{Prop}
Compatibility of the linear system~\eqref{eq:ad-lin-bil-AB} with the constraint~\eqref{eq:int-P-lin-AB} implies the equations
\begin{equation} 
A_i(n+e_m) B_i(n) = B_i(n-e_i) A_i(n), \quad i = 1,\dots , m-1,
\end{equation}
and
\begin{equation} \label{eq:P-A-B}
x_{|n|+m} - x_{|n|+m-1}= B_i(n-e_i) \Sigma(n-e_i) - \Sigma(n-e_m) B_i(n-e_i-e_m), \qquad i=1, \dots , m-1,
\end{equation}
where
\begin{equation}
\Sigma(n) = \frac{A_1(n)}{B_1(n)} + \dots +\frac{A_{m-1}(n)}{B_{m-1}(n)} + 1.
\end{equation}
\end{Prop}
\begin{Rem}
One can check that in terms of the notation of the previous Section
\begin{equation}
\Sigma(n) = \frac{S(n)}{\tau(n-e_m) \tau(n+e_m)},
\end{equation}
and equation \eqref{eq:P-A-B} is then equivalent to \eqref{eq:nlin-tau-S} for the index $j=m$.
\end{Rem}
\begin{Rem}
We presented the basic set within the system of integrable nonlinear partial difference equations derived from compatibility of the linear system \eqref{eq:ad-lin-bil-AB} and \eqref{eq:int-P-lin-AB}. Other equations, like for example the version of equation 
\eqref{eq:nlin-tau-S} for $i$ and $j$ different from the distinguished variable $m$, reads in the present notation 
\begin{equation*}
x_{|n|+m} - x_{|n|+m-1}= 
\frac{B_i(n-e_i) B_j(n-e_i-e_j) \Sigma(n-e_i) - 
B_j(n-e_j)  B_i(n-e_i-e_j) \Sigma(n-e_j)}
{B_j(n-e_j) -B_i(n-e_i) } ,
\end{equation*}
and can be derived from other equations of this Section.
\end{Rem}

\section{The two-dimensional subcase} \label{sec:2D}
Below we discuss in more detail the simplest case $m=2$ and we make connection to previous works. To elucidate the connection we change the notation to more explicit one, used in this context, putting the discrete variables as subscripts. For example the linear problem for the non-autonomous discrete-time Toda lattice equation reads
\begin{align}
\label{eq:ad-lin-bil-m2-ex}
\bphi_{n_1+1,n_2+1}(x) \tau_{n_1,n_2} & =	
\bphi_{n_1,n_2+1}(x) \tau_{n_1+1,n_2}
- \bphi_{n_1+1,n_2}(x) \tau_{n_1,n_2+1} , \\
\label{eq:P-lin-m2-ex}
(x-x_{n_1 + n_2 + 2})\bphi_{n_1,n_2} (x) \tau_{n_1,n_2}  & = \bphi_{n_1+1,n_2}(x) \tau_{n_1-1,n_2} +  \bphi_{n_1,n_2+1}(x) \tau_{n_1,n_2-1},	
\end{align}
while the non-linear equations are of the form
\begin{equation} \label{eq:nlin-tau-S-2D-ex}
x_{n_1 + n_2 + 2}-x_{n_1 + n_2 + 1} = \frac{ \tau_{n_1-1,n_2}\tau_{n_1,n_2-1}}
{\tau_{n_1-1,n_2-1} \tau_{n_1,n_2}} \left( \frac{S_{n_1-1,n_2}}{\tau^2_{n_1-1,n_2}} -
\frac{S_{n_1,n_2-1}}{\tau^2_{n_1,n_2-1}} \right), 
\end{equation}
where
\begin{equation} \label{eq:S-2D}
S_{n_1,n_2} = \tau_{n_1+1,n_2} \tau_{n_1-1,n_2} +  \tau_{n_1,n_2+1}\tau_{n_1,n_2-1}.
\end{equation}

\subsection{Non-autonomous discrete-time Toda equation}
In the  case $m=2$ the linear system \eqref{eq:ad-lin-bil-AB}-\eqref{eq:int-P-lin-AB} takes the following form 
\begin{align}
\bpsi_{n_1 +1, n_2+1}(x) = & \bpsi_{n_1,n_2+1}(x) - \bpsi_{n_1+1,n_2}(x) B_{n_1,n_2},\\
(x-x_{n_1+n_2 +2}) \bpsi_{n_1.n_2}(x) = & \bpsi_{n_1+1,n_2}(x) A_{n_1,n_2} + \bpsi_{n_1,n_2+1}(x).
\end{align}
The first part \eqref{eq:BB}-\eqref{eq:B+B} of the compatibility condition is empty and the non-linear equations reduce to
\begin{align}
A_{n_1,n_2} B_{n_1 - 1,n_2} & = A_{n_1,n_2+1}B_{n_1,n_2} \\
x_{n_1+n_2 + 2} - x_{n_1+n_2+1} & = A_{n_1-1,n_2} + B_{n_1-1,n_2} - A_{n_1,n_1} - B_{n_1-1,n_2-1}.
\end{align}

The following change of variables
\begin{equation*}
t = -(n_1+n_2), \qquad k = n_1, \qquad C^t_k = B_{n_1,n_2}, \qquad D^{t+1}_k = - A_{n_1,n_2}, \qquad \lambda_{t+1} = x_{n_1+n_2},
\end{equation*}
gives the non-autonomous discrete-time Toda lattice equations of the form
\begin{align}
C^{t+1}_{k-1} D^{t+1}_k & = C^t_k D^t_k , \\
\lambda_{t+1} + C^{t+1}_{k} + D^{t+1}_{k} & =
\lambda_{t} + C^{t}_{k} + D^{t}_{k+1},
\end{align}
studied in~\cite{SpiridonovZhedanov,Hirota-naTl,KajiwaraMukaihira, MukaihiraTsujimoto}.

\subsection{The rational interpolation of a function}

When $R_{n_1,n_2}(x) = P_{n_1,n_2}(x)/Q_{n_1,n_2}(x)$ is the rational interpolation of a function $f(x)$ with nodes in points $x_s$, where $s=1, \dots N$, and $N=n_1 + n_2 + 1$, then the interpolation condition can be brought to the form
\begin{equation}
- P_{n_1,n_2}(x_s) + Q_{n_1,n_1}(x_s) f(x_s) = 0, \quad s=1, \dots N. 
\end{equation}
The determinant representation of the nominator reads 
\begin{equation}
P_{n_1,n_2}(x) = \left| \begin{matrix}
-1 & -x_1 & \dots & -x_1^{n_1} & f(x_1) & x_1 f(x_1) & \dots & x_1^{n_2} f(x_1) \\
\vdots & \vdots & & \vdots & \vdots & \vdots & & \vdots \\
-1 & -x_N & \dots & -x_N^{n_1} & f(x_N) & x_N f(x_N) & \dots & x_N^{n_2} f(x_N) \\
1 & x & \dots & x^{n_1} & 0 & 0 & \dots & 0
\end{matrix} \right|
\end{equation}
and can be calculated using the generalized Laplace expansion with respect to the first $n_1 + 1$ columns. The corresponding sum consists of expressions which up to a sign have the form
\begin{equation*}
\begin{split}
\left| \begin{matrix}
1 & x & \dots & x^{n_1} \\
1 & x_{i_1} & \dots & x_{i_1}^{n_1} \\
\vdots & \vdots & & \vdots \\
1 & x_{i_{n_1}}& \dots & x_{i_{n_1}}^{n_1} 
\end{matrix} \right| 
\left| \begin{matrix}
f(x_{j_1}) & 0 & \dots & 0\\
0 & f(x_{j_2}) & \dots & 0 \\
\vdots & \vdots & \ddots & \vdots \\
0 & 0 & \dots & f(x_{j_{n_2+1}})
\end{matrix} \right| 
\left| \begin{matrix}
1 &x_{j_1}  & \dots & x_{j_1}^{n_2} \\
1 &x_{j_2}  & \dots & x_{j_2}^{n_2}  \\
\vdots & \vdots & & \vdots \\
1 &x_{j_{n_2+1}} & \dots & x_{j_{n_2+1}}^{n_2}
\end{matrix} \right|
 = \\=
\prod_{1\leq k \leq n_1} (x-x_{i_k})
\prod_{1 \leq k<\ell \leq n_1} (x_{i_k}-x_{i_\ell}) 
\prod_{1\leq p \leq n_2+1} f(x_{j_p})
\prod_{1 \leq p < q \leq n_2+1} (x_{i_p}-x_{i_q}) ,\qquad 
\end{split}
\end{equation*}
where
\begin{gather*}
1 \leq i_1 < i_2 < \dots < i_{n_1} \leq N, \quad
1 \leq j_1 < j_2 < \dots < j_{n_2+1} \leq N, \\
\{  i_1, i_2,  \dots , i_{n_1} \} \cap \{  j_1, j_2, \dots , j_{n_2+1} \} = \emptyset .
\end{gather*}
Putting them together we obtain
\begin{equation}
P_{n_1,n_2}(x) = (-1)^{\sigma(n_1)} \sum_{\begin{smallmatrix} I \subset \{ 1, \dots ,N\}\\ |I| = n_1 \end{smallmatrix}} (-1)^{S(I)} \prod_{i\in I} (x-x_i) \prod_{\begin{smallmatrix} i<j \\i,j\in I
	\end{smallmatrix}}(x_i - x_j)
\prod_{s\in \bar I} f(x_s) \prod_{\begin{smallmatrix} s<t \\ s,t\in \bar{I}\end{smallmatrix}} (x_s - x_t),
\end{equation}
where
\begin{equation*}
\sigma(n_1) = \frac{(n_1 + 1)(n_1 + 2)}{2},  \qquad S(I) =  i_1 + i_2 +  \dots + i_{n_1}, \quad \text{for} \quad I = \{  i_1, i_2,  \dots , i_{n_1} \},
\end{equation*}
and $ \bar{I}$ is the complement of $I$ in $\{ 1,2,\dots , N\}$. Similarly, by the generalized Laplace expansion of the determinant representing the denominator
\begin{equation}
Q_{n_1,n_2}(x) = \left| \begin{matrix}
-1 & -x_1 & \dots & -x_1^{n_1} & f(x_1) & x_1 f(x_1) & \dots & x_1^{n_2} f(x_1) \\
\vdots & \vdots & & \vdots & \vdots & \vdots & & \vdots \\
-1 & -x_N & \dots & -x_N^{n_1} & f(x_N) & x_N f(x_N) & \dots & x_N^{n_2} f(x_N) \\
0 & 0 & \dots & 0  & 1 & x & \dots & x^{n_2}
\end{matrix} \right|
\end{equation}
we obtain
\begin{equation}
Q_{n_1,n_2}(x) = (-1)^{\sigma(n_1)+N} \sum_{\begin{smallmatrix} J \subset \{ 1, \dots ,N\}\\ |J| = n_1+1 \end{smallmatrix}} (-1)^{S(J)} \prod_{\begin{smallmatrix} i<j \\i,j\in J
	\end{smallmatrix}}(x_i - x_j)
\prod_{s\in \bar{J}}  (x-x_s) f(x_s) \prod_{\begin{smallmatrix} s<t \\ s,t\in \bar{J}\end{smallmatrix}} (x_s - x_t),
\end{equation}
what concludes derivation of a division free version of the classical rational interpolation formulas by Cauchy~\cite{Cauchy}.

 Finally let us present the corresponding expression for the basic determinants which provide solutions to the non-autonomous discrete-time Toda lattice equations in the form~\eqref{eq:nlin-tau-S-2D-ex}-\eqref{eq:S-2D}
\begin{equation} \begin{split}
\Delta_{n_1,n_2} = & \left| \begin{matrix}
-1 & -x_1 & \dots & -x_1^{n_1} & f(x_1) & x_1 f(x_1) & \dots & x_1^{n_2} f(x_1) \\
\vdots & \vdots & & \vdots & \vdots & \vdots & & \vdots \\
-1 & -x_{N+1} & \dots & -x_{N+1}^{n_1} & f(x_{N+1}) & x_{N+1} f(x_{N+1}) & \dots & x_{N+1}^{n_2} f(x_{N+1}) 
\end{matrix} \right| = \\ 
= & (-1)^{\sigma(n_1)+n_1 + 1} \sum_{\begin{smallmatrix} K \subset \{ 1, \dots ,N +1 \}\\ |K| = n_1+1 \end{smallmatrix}} (-1)^{S(K)} \prod_{\begin{smallmatrix} i<j \\i,j\in K
	\end{smallmatrix}}(x_i - x_j)
\prod_{s\in \bar{K}}  f(x_s) \prod_{\begin{smallmatrix} s<t \\ s,t\in \bar{K}\end{smallmatrix}} (x_s - x_t).
\end{split}
\end{equation}

\subsection{Non-autonomous Wynn recurrence}
In the limiting case of the Pad\'{e} approximation the ratio of two solutions of the linear problem satisfies the so called Wynn recurrence~\cite{Wynn,BakerGraves-Morris}; see also~\cite{Doliwa-Siemaszko-W}. The rational interpolation analog of the recurrence is known in the literature~\cite{Claessens,CuytWuytack}. Let us present its derivation within the present formalism. The result is valid in a broader context of integrable systems theory and is not restricted to the interpolation applications only.
\begin{Prop}
	Given two non-trivial solutions $P_{n_1,n_2}(x)$ and $Q_{n_1,n_2}(x)$ of the two-dimensional version~\eqref{eq:ad-lin-bil-m2-ex}-\eqref{eq:P-lin-m2-ex} of the linear problem
then their ratio $R_{n_1,n_2}(x) = P_{n_1,n_2}(x)/Q_{n_1,n_2}(x)$ satisfies the following non-autonomous generalization of Wynn's cross-rule
\begin{equation*}
\frac{x-x_{n_1 + n_2 +2}}{R_{n_1-1,n_2}(x) - R_{n_1,n_2}(x)} +
\frac{x-x_{n_1 + n_2 +1}}{R_{n_1+1,n_2}(x) - R_{n_1,n_2}(x)} =
\frac{x-x_{n_1 + n_2 +2}}{R_{n_1,n_2-1}(x) - R_{n_1,n_2}(x)} +
\frac{x-x_{n_1 + n_2 +1}}{R_{n_1,n_2+1}(x) - R_{n_1,n_2}(x)}.
\end{equation*}
\end{Prop}
\begin{proof}
	Summing up the inverse of the direct consequence of \eqref{eq:ad-lin-bil-m2-ex}
\begin{equation*}
R_{n_1-1,n_2}(x) - R_{n_1,n_2}(x) = \frac{\tau_{n_1-1,n_2-1}[P_{n_1-1,n_2+1}(x) Q_{n_1,n_2}(x) - Q_{n_1-1,n_2+1}(x) P_{n_1,n_2}(x)]}{Q_{n_1,n_2}(x)[
Q_{n_1,n_2}(x)\tau_{n_1-2,n_2} + Q_{n_1-1,n_2+1}(x) \tau_{n_1-1,n_2-1}]}
\end{equation*}
with the inverse of the following consequence of both equations \eqref{eq:ad-lin-bil-m2-ex} and \eqref{eq:P-lin-m2-ex}
\begin{equation*}
\begin{split}
&R_{n_1,n_2}(x) - R_{n_1,n_2-1}(x) = \\&
\frac{\tau_{n_1-1,n_2-1}
	\left[P_{n_1-1,n_2+1}(x) Q_{n_1,n_2}(x) - 
	Q_{n_1-1,n_2+1}(x) P_{n_1,n_2}(x)\right]}
{Q_{n_1,n_2}(x)\left[
	Q_{n_1,n_2}(x)\left( (x-x_{n_1+n_2+1})\tau_{n_1-1,n_2}\tau_{n_1-1,n_2-1}/
	\tau_{n_1,n_2-1}  -\tau_{n_1-2,n_2} \right) - Q_{n_1-1,n_2+1}(x)  \tau_{n_1-1,n_2-1}\right]}
\end{split}
\end{equation*}
we obtain
\begin{equation*}
\frac{1}{R_{n_1-1,n_2}(x) - R_{n_1,n_2}(x)} + \frac{1}{R_{n_1,n_2}(x) - R_{n_1,n_2-1}(x) } = 
\frac{ (x-x_{n_1+n_2+1})\frac{Q_{n_1,n_2}(x)}{Q_{n_1-1,n_2+1}(x)} \frac{\tau_{n_1-1,n_2}}{\tau_{n_1,n_2-1}}}
	{R_{n_1-1,n_2+1}(x) - R_{n_1,n_2}(x) }.
\end{equation*}
Analogous calculation of the expression
\begin{equation*}
\frac{1}{R_{n_1,n_2+1}(x) - R_{n_1,n_2}(x)} + \frac{1}{R_{n_1,n_2}(x) - R_{n_1+1,n_2}(x) } = 
\frac{ (x-x_{n_1+n_2+2})\frac{Q_{n_1,n_2}(x)}{Q_{n_1-1,n_2+1}(x)} \frac{\tau_{n_1-1,n_2}}{\tau_{n_1,n_2-1}}}
{R_{n_1-1,n_2+1}(x) - R_{n_1,n_2}(x) }
\end{equation*}
concludes the proof.
\end{proof}
 \begin{Rem}
 	We leave to the interested Reader derivation of the non-autonomous Wynn recurrence using the formalism of Section~\ref{sec:AB}; see also~\cite{Claessens}.
 \end{Rem}
\begin{Rem}
	The non-autonomous Wynn recurrence is the simplest case of cross-equations, whose multidimensional consistency on face centered cubic lattice was studied in detail in~\cite{Kels}. 
\end{Rem}

\section{Conclusion and open problems}

We presented new non-autonomous multidimensional (in the sense of the number of variables involved) integrable system which in two dimensions reduces to the non-autonomous discrete-time Toda lattice equations. In deriving the system we were guided by the interpolation generalization of the Hermite--Pad\'{e} approximation problem. We presented also the corresponding class of solutions of the system in terms of determinants built from the interpolation data. 

It would be instructive to apply to the system standard techniques of the integrable systems theory, like the direct methods of construction of the soliton solutions~\cite{Hirota-book,KajiwaraMukaihira}, the inverse spectral transform to find more general solutions~\cite{AblowitzSegur,Konopelchenko-book}, the algebro-geometric techniques to generate finite-gap quasi-periodic solutions~\cite{BBEIM}, the Darboux transformation~\cite{Matveev}. 
However more challenging problems in the context of the system, as we can see, are:
	clarification of its connection with the theory of discrete multiple orthogonal polynomials~\cite{DOP,MDOP,NikiforovSuslovUvarov}, 
	a possible extension of the theory to the multivariate case as suggested by the recent work~\cite{AriznabarretaManas} on the multivariate Toda hierarchy,
construction of its non-commutative and quantum generalizations~\cite{Doliwa-NCHP},
description of its symmetry structure in terms of root lattices and Weyl groups~\cite{Noumi-book,Dol-AN},
	systematic symmetry reductions to lower dimensional equations~\cite{NagaoYamada},
construction of its version with self-consistent sources~\cite{DoliwaLin},
	finding its geometric interpretation in terms of lattice submanifolds~\cite{Dol-Des,Doliwa-Siemaszko-HP}.
Any of them deserves a~separate detailed study. 


\providecommand{\bysame}{\leavevmode\hbox to3em{\hrulefill}\thinspace}

\end{document}